\documentclass[a4paper,11pt,preprint]{article}

\usepackage{amsmath}
\usepackage{amsthm}
\usepackage{amssymb}

\usepackage{algorithmic}
\usepackage{algorithm}
\usepackage{xcolor}
\usepackage{graphicx}
\usepackage{subfig}

\usepackage{wrapfig}

\usepackage{multicol}
\usepackage{units}

\usepackage{cite}

\usepackage{paralist}

\usepackage{CJKutf8}
\usepackage{comment}
\theoremstyle{definition}

\newtheorem{definition}{Definition}
\newtheorem{Property}{Property}
\newtheorem{theorem}{Theorem}
\newtheorem{corollary}[theorem]{Corollary}
\newtheorem{lemma}[theorem]{Lemma}

\graphicspath{{figure/}}

\newcommand{\BIGLR}[3]{{\left#1#3\right#2}}
\newcommand{\BIGP}[1]{{\BIGLR{(}{)}{#1}}}

\begin{document}
\begin{CJK*}{UTF8}{bsmi}

\title{\textbf{Optimal Construction for Time-Convex Hull with Two Orthogonal Highways in the $L_1$-metric}}

\author{Jyun-Yu Chen\textsuperscript{1}  , Po-Hsuan Chen\textsuperscript{2} \\
\textsuperscript{1} Department of Computer Science and Information Engineering \\
National Chung Cheng University, Chiayi, Taiwan\\
tchcjy@gmail.com\\
\textsuperscript{2} Department of Computer Science and Information Engineering\\
National Chung Cheng University, Chiayi, Taiwan\\
ramusa19@alum.ccu.edu.tw \\
}

\date{}

\maketitle

\begin{abstract}
	
\hspace{4mm}We consider the time-convex hull problem in the presence of two orthogonal highways \textbf{H}.
In this problem, the travelling speed on the highway is faster than off the highway, and the time-convex hull of a point set \textbf{P} is the closure of \textbf{P} with
respect to the inclusion of shortest time-paths.

	In this paper, we provide the algorithm for constructing the time-convex hull with two orthogonal highways. We reach the optimal result of $O(n\log{n})$ time for arbitrary highway speed in the L1-metric. For the L2-metric with infinite highway speed, we hit the goal of $O(n\log{n})$ time as well.

\end{abstract}
%


\section{Introduction}

\hspace{4mm} Path planing of a transition network has been an important problem in recent years. In a complex transition network, lots of models have been provided to solve traffic issues. Highway is one of a simple model in it. We assume that we can enter or exit the highway at any point, and the moving speed on the highway is $v>1$, while the speed off the highway is $v=1$. Due to the highway, the path of two points using highway may be faster than their straight-line path. Therefore, we care about the minimum travelling time of a path rather than its distance. To replace the definition of distance, we make use of the measure named \textbf{time-distance}. The time-distance of two points is defined to be the minimum travelling time from one point to the other one. 
	
	The convex hull in the presence of a highway was introduced by Hurtado et al \cite{Hurtado}. A set $S$ is convex if it contains the shortest time-path between any two point in $S$. We define the time-convex hull of a set $S$, which is a closure of $S$ with respect to the inclusion of shortest time-paths.\bigskip
		
	In previous works, the convex hull in the presence of highway was first solved by Palop \cite{Palop}. Palop has showed that, the convex hull in the presence of highway is composed of convex polygons(clusters) with segments of highway connecting all the components. According to the definition of clusters, the shortest time-path of any two points which belongs to different clusters must use the highway. Therefore, we can comprehend the degree of convenience in the transportation network from the point density of each clusters.\bigskip 
	
	Palop \cite{Palop} provided a $\Theta({n^{2}})$ time algorithm for constructing the convex hull in the presence of a highway by enumerating the shortest time-path between two points. Then, Yu and Lee \cite{Yu} studied the problem and provided an $O(n\log{n})$ approach based on incremental point insertions. However, the algorithm which they proposed does not return the correct answer in all circumstances. Some critical cases which make the problem more tough were overlooked. After that, Aloupis et al. \cite{Aloupis} proposed the algorithm that takes $O(n\log^{2}{n})$ time for $L_2$-metric and takes $O(n\log{n})$ time for $L_1$-metric. They solved the cases which Yu and Lee \cite{Yu} overlooked. Dai et al. \cite{Dai} provide the $O(n\log{n})$ algorithm for $L_p$-metric, where 1 $\leq p \leq \infty$, and they reduce the cluster-merging step which was proposed by Aloupis et al. \cite{Aloupis} to a geometric query.

\begin{flushleft}
 \textbf{Our focus and contribution}
\end{flushleft}

	We provide the algorithm for computing time-convex hull with two orthogonal highways in two different distance metrics. First, we provide the algorithm for $L_1$-metric with arbitrary highway speed and compute it in $O(n\log{n})$ time. In the algorithm which is given by Aloupis et al. \cite{Aloupis}, they proposed the approach to constructing the time convex hull with one highway. Considering that there are two orthogonal highways in our model, we divide the input point set into two parts and use the approach in each part. However, after completing this work, the clusters may need to be merged between two parts. Therefore, we use the approach given by Aloupis et al. \cite{Aloupis} as a subroutine in our algorithm and apply the data structure proposed by Mitchell \cite{Mitchell} to merge the clusters between two parts. All together the above processes generate our $O(n\log{n})$-algorithm for $L_1$-metric.
	
	Second, we provide the $O(n\log{n})$-algorithm for the special case in $L_2$-metric where the highway speed is infinite. Our $O(n\log{n})$-algorithm for $L_2$-metric with infinite highway speed include the following steps: Like what we do in $L_1$-metric, we divide the input point set into two parts and construct clusters with the algorithm proposed by Dai et al. \cite{Dai} in each part. Dai et al. use the symmetric property of cluster-merging condition and reduce the time-convex hull problem to the geometric query. Thus, the data structure proposed by Mitchell \cite{Mitchell} could be used to answer this particular geometric problem. Finally, after clusters are constructed, we use the ray-shooting method to check if any point should be merged with other cluster.

\section{Notation and definition}
\hspace{4mm} In this section, we will present some definitions that are useful for the following sections.

\begin{definition}[$L_p$-metric] \label{Lp-metric}
for any two points a($a_1$,$a_2$,...,$a_n$), b($b_1$,$b_2$,...,$b_n$) under $L_p$-metric, where p $\in$ $Z^{+}$, their distance is defined as $\sqrt[p]{\sum_{i=1}^{n}|a_i - b_i|^p}$.
\end{definition}

	When the metric are different, the distance of two points will also be different. Therefore, when we are solving the problems about distance, it's important to make sure that which metric it is.

\begin{definition}[Highway] \label{Highway}
The travelling speed $V_H$ on a highway \emph{H} in $\mathbb{R}^{2}$ is faster than 1, while the travelling speed without highways is 1. In this paper, there are two highways intersect perpendicularly. 
\end{definition}

	Without loss of generality, we position highways as x-axis and y-axis, which are denoted $H_x$ and $H_y$. Furthermore, because the properties of every quadrants are the same, we could concern only the first quadrant.

\begin{definition}[Side] \label{Side}
Split the space into two pieces by the function$:x=y$. The one closer to $H_x$ is called $Side-H_x$, and the one closer to $H_y$ is called $Side-H_y$.
\end{definition}

\begin{figure}
\begin{center}
	\includegraphics[scale=0.25]{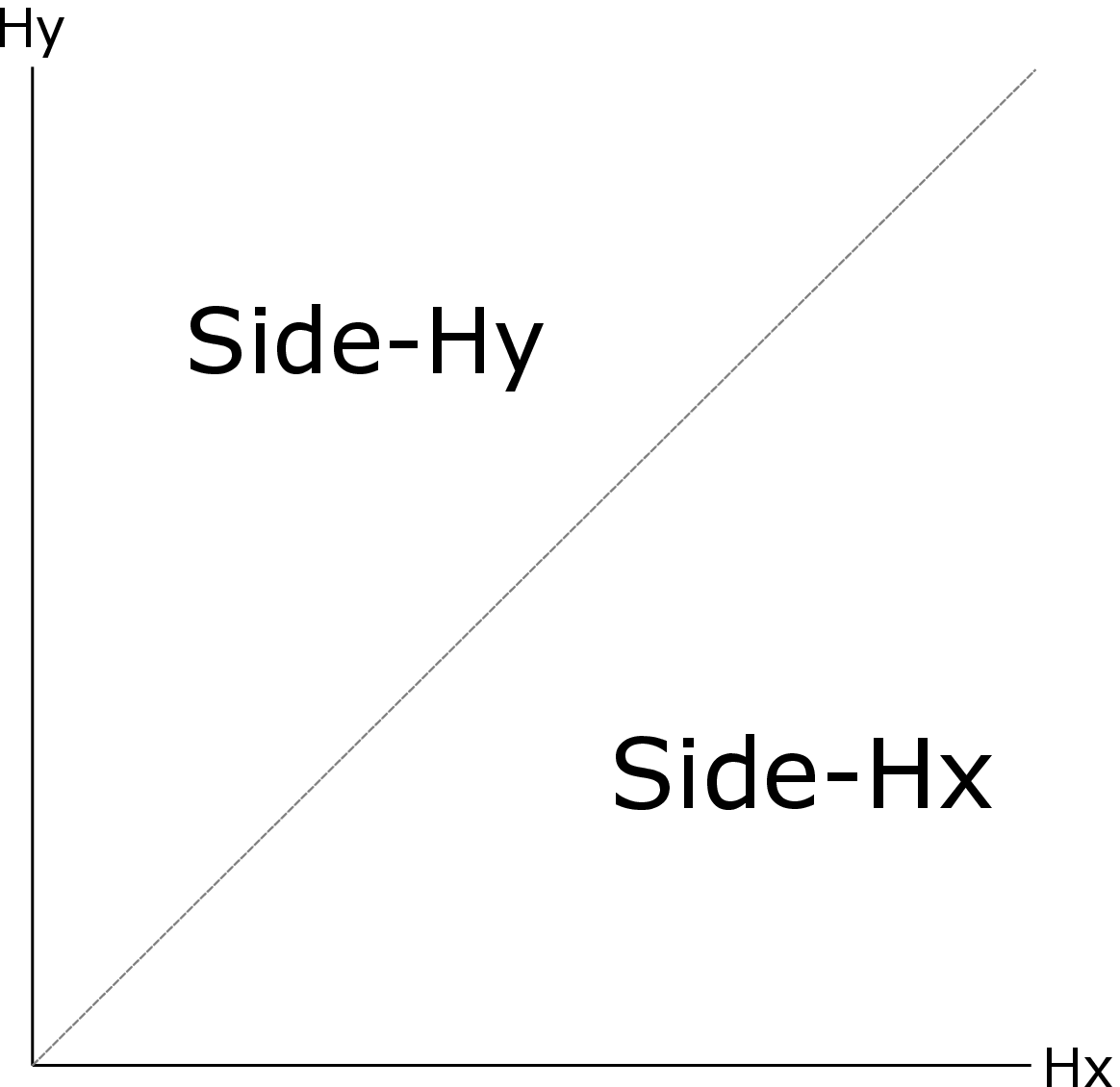} \label{fig:side}
	\caption{Sides}
\end{center}
\end{figure}

	With highways, the traffic time between points may less than not using highways. To find the fastest path, we should check that how much time it takes instead of how far it is. In other word, the measure of a path is not length but time, which is called \textbf{time distance}.

\begin{definition}[Time distance] \label{Time distance}
For any two points p and q, their time distance is the travelling time that from one point to the other one. The time distance of p, q is denoted d(p, q). \bigskip

	For any two points $p$ and $q$, let STP($p$,$q$) be the set of their shortest time-paths, which may either use highway or not. If STP(p,q) $\cap$ $H$ $\neq \varnothing $, then the shortest time-path of p,q use highway.
\end{definition}

	However, the shortest time-path of two points under $L_1$-metric is not unique. Thus, we define the shortest time-path between any two points to be the path closest to the highway. In other word, the shortest time-path in $L_1$-metric would be "L shape"(See Fig. \ref{fig:L_shape}).

\begin{figure}[H] 
\begin{center}
	\includegraphics[scale=0.6]{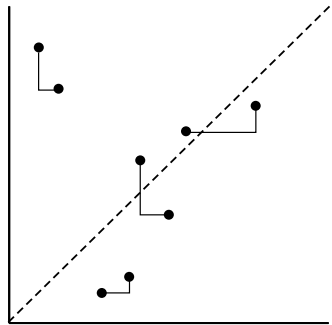}
	\caption{Shortest path under $L_1$-metric is L shape} \label{fig:L_shape}
\end{center}
\end{figure}

\begin{definition}[Shortest time-path in $L_1$-metric] \label{STP L1}
	Any two points not using highway will choose the path which is the closest to the highway.
\end{definition}

\textit{Convex hull and Time-convex hull} : In general definition, the convex hull of a point set Q includes all the shortest paths of any pair in the point set Q. In our model, considering time distance, a point set Q is said to be time-convex if it contains all shortest time-paths of every pairs in Q.\bigskip

\begin{definition}[Time-convex hull] \label{TCH}
The time-convex hull of a point set P, denoted TCH(P), is a minimum time-convex set of any pairs in P. In other word, for any two points $p_1, p_2 \in P$, STP($p_1, p_2$) would located in TCH(P).
\end{definition}

The time-convex hull is composed of two parts : two orthogonal straight-line highways, and the set of convex connected components which is denoted by clusters. To be brief, the polyhedrons of the time-convex hull of the point set Q.

\begin{definition}[Cluster] \label{cluster}
For a point set P, clusters are defined to be the connected components of TCH(P) $\backslash$ P.
\end{definition}

\begin{figure*}
\begin{center}
	\includegraphics[scale=0.7]{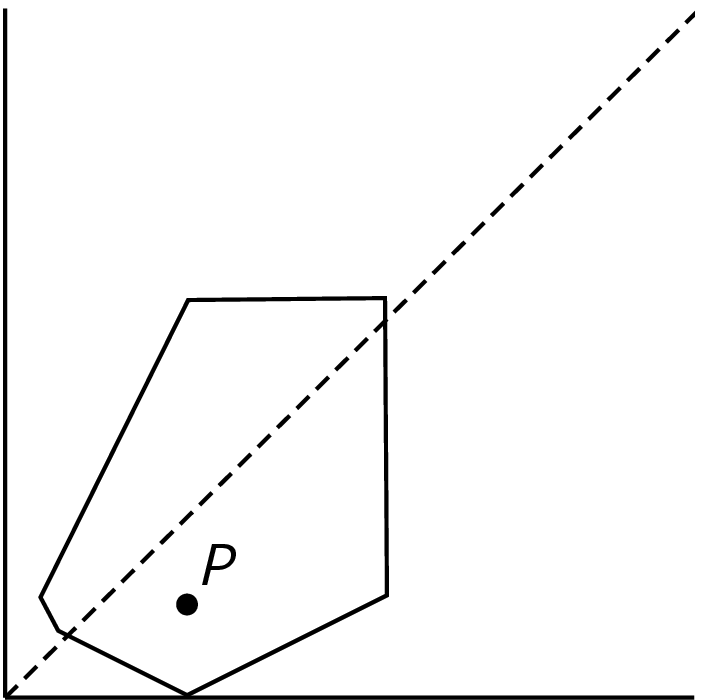}
	\includegraphics[scale=0.7]{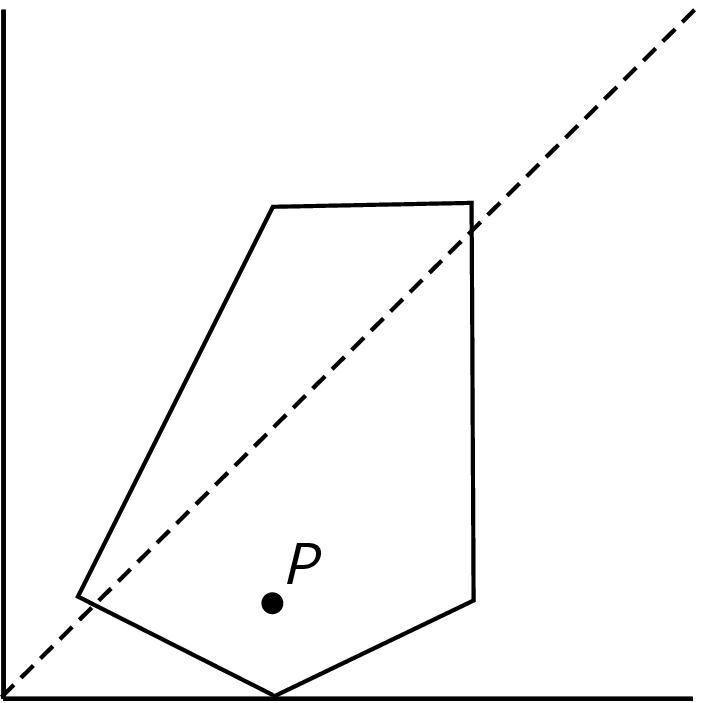} 
	\includegraphics[scale=0.7]{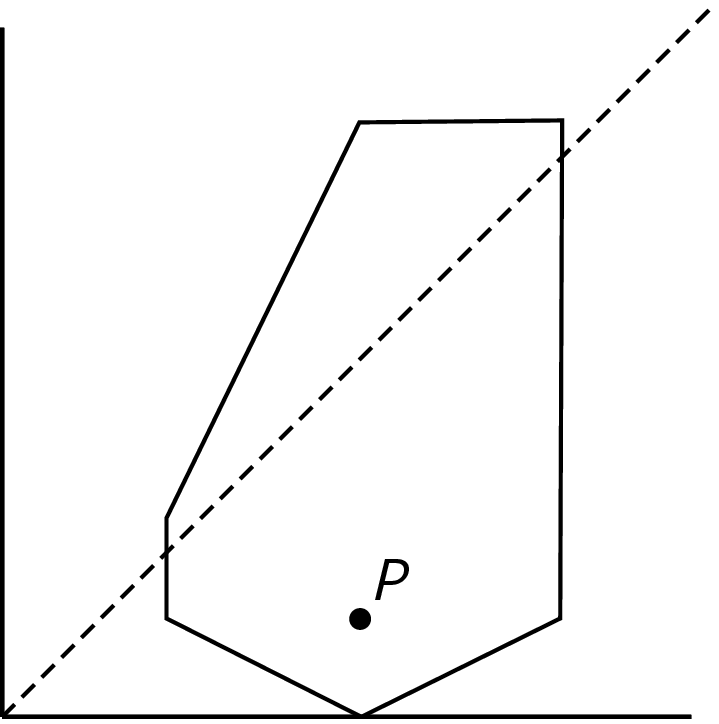}
	\caption{Walking regions under $L_1$-metric} \label{fig:WR L1}
\end{center}
\end{figure*}

\section{Walking Region}

\hspace{4mm} For any two points p, q, if the shortest time-path of p, q not use highway, or we can say that, q lies in the \textbf{walking region} of p.

\begin{definition}[Walking Region] \label{WR}
The walking region of a point p, denoted WR(p), is the set of points q such that STP(p, q) does not use highway $H$.$\lbrace q\in\mathbb{R}^{2} :|pq| \leq  STP(p,q)\cap H \neq \varnothing  \rbrace$ 


\end{definition}

	In the research of Aloupis et al. \cite{Aloupis}, they proposed the boundary of walking region of points with one highway. According to definition \ref{WR}, we know that the walking region boundary is yielded by comparing the shortest path using highway with the shortest path without highway. In our model, because there are two orthogonal highways, the point $p$ can enter from $H_x$ or $H_y$ then out of $H_x$ or $H_y$. Thus, there are 4 kinds of the shortest paths using highway (denoted highway path). We focus on first quadrant, we consider only the points on $Side-H_x$ since the walking regions of the points on $Side-H_y$ are symmetric with respect to $x=y$. For any point $p$ $\epsilon$ S, let $p'_x$ be its orthogonal projection onto $H_x$, $p'_y$ be its orthogonal projection onto $H_y$, and the intersection point of two highways $O$. The shortest paths between two points $p$ and $q$ consists of either:\\
\begin{compactenum}
	\item[-] The $L_1$ distance of $p$ and $q$.
	\item[-] The horizontal segment $p'_x$ $q'_x$ and the vertical segments $pp'_x$ and $qq'_x$. 
	\item[-] The horizontal segments $Op'_x$ and $qq'_y$ and the vertical segments $pp'_x$ and $Oq'_y$.
	\item[-] The horizontal segments $pp'_y$ and $qq'_y$ and the vertical segment $p'_y$ $q'_y$.
	\item[-] The horizontal segments $pp'_y$ and $Oq'_y$ and the vertical segments $Op'_y$ and $qq'_x$.
\end{compactenum}

	The walking region of $a$ is the set $\lbrace b\in\mathbb{R}^{2}$ : $L_1$ distance of a and b $ \leq $ every kinds of highway paths$\rbrace$.

\begin{Property}[Walking Region in $L_1$] \label{WR_L1} (See \cite{Aloupis}, \cite{Palop})
For any point q = ($x_q$, $y_q$) with $x_q,y_q \geq 0$ , $x_q \geq y_q$ , we have the following properties : 
\end{Property}
 In the $L_1$-metric, there are 3 kinds of shape (See figure \ref{fig:WR L1}) related to the x-coordinate and y-coordinate of $q$. The walking region of the point $q$ is the intersection of four regions :
\begin{compactenum}
	\item
		 two segments with slope $\dfrac{(1-1/V_H)}{2}$ and $-\dfrac{(1-1/V_H)}{2}$ intersect at $(x_q,0)$ and connect the half-line at horizontal line through $q$.
		

	\item
		 two segments with slope $\dfrac{(1-1/V_H)}{2}$ and $-\dfrac{(1-1/V_H)}{2}$ intersect at $(x_q ,\dfrac{(x_q - y_q)}{2} + \dfrac{(x_q + y_q)}{2V_H} )$ and connect the half-line at horizontal line through $q$.

	\item
		two segments with slope $\dfrac{2}{(1-1/V_H)}$ and $-\dfrac{2}{(1-1/V_H)}$ which meet at $(0, y_q)$ and join the half-line at vertical line through $q$.
		
	\item
		two segments with slope $\dfrac{2}{(1-1/V_H)}$ and $-\dfrac{2}{(1-1/V_H)}$ which meet at $((x_q-y_q)/2 + (y_q - x_q)/ 2V_H, y_q)$ and join the half-line at vertical line through $q$.\bigskip
\end{compactenum}

\begin{proof}
	For any point $p$, the 4 kinds of regions are the boundary of the set $\lbrace q\in\mathbb{R}^{2} :|pq| \leq $ 4 kinds of highway path $\rbrace$. For example, the first region is the set $\lbrace q\in\mathbb{R}^{2} :|pq| \leq |pp_x'| + |p_x'q_x'|/v + |qq_x'| \rbrace$. In other word, for any point $q$ in this region, the shortest path between $p$ and $q$ which does not use the highway is faster than using this highway path. Therefore, for the point $q$ which is in the intersection of this four regions, the shortest time-path between $p$ and $q$ do not use the highway.

\end{proof}

\begin{Property}[Walking region in $L_2$] \label{WR_L2} (See \cite{Dai}, \cite{Palop})
For any point q = ($x_q$, $y_q$) with $x_q,y_q \geq 0$ , $x_q \geq y_q$ , we have the following properties :
\end{Property}
In the $L_2$-metric with infinite highway speed, the boundary of WR($p$) is the intersection of two regions :
\begin{compactenum}
	\item
		The first region is characterized by the following two parabolas :\\
		(a) \emph{right discriminating parabola}, which is the curve satisfying 
				\begin{align*}
				\begin{cases}
					x\ge x_q + y_q\tan\alpha, \quad \text{and} \\
					\sqrt{\BIGP{x-x_q}^2+\BIGP{y-y_q}^2} = y_q\sec\alpha + y\sec\alpha %
				\end{cases}
				\end{align*}
		(b) The \emph{left discriminating parabola} 
		which is symmetric to the right discriminating parabola with respect to the line $x=x_q$. 

	\item
		The second region is characterized by the following two parabolas :\\
		(a) \emph{right discriminating parabola}, which is the curve satisfying 
				\begin{align*}
				\begin{cases}
					y\ge y_q - x_q\tan\alpha, \quad \text{and} \\
					\sqrt{\BIGP{x-x_q}^2+\BIGP{y-y_q}^2} = y_q\sec\alpha + x\sec\alpha %
				\end{cases}
				\end{align*}
		(b) The \emph{left discriminating parabola} 
		which is symmetric to the right discriminating parabola with respect to the line $x=x_q$. 
\end{compactenum}

\begin{figure}[H]
\begin{center}
	\includegraphics[scale=0.3]{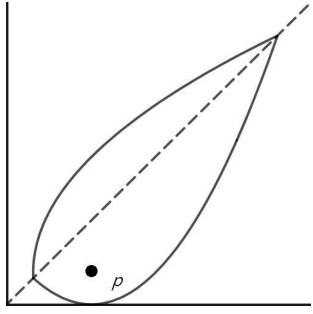} 
	\caption{Walking regions under $L_2$-metric with infinite $V_H$} \label{fig:WR L2}
\end{center}
\end{figure}

\begin{proof}
	It is similar to the case in $L_1$-metric. In $L_2$-metric, the shortest path of  the point in $Side-H_x$ always enter the highway $H_x$ because of the infinite highway speed. So, there are just 2 kinds of highway paths and regions.

\end{proof}

\section{Optimal Construction in $L_1$-metric}
\hspace{4mm} In our model, there exist two orthogonal highways and a point set. The two orthogonal highways could be regard as x-axis and y-axis.

	Image that there is another line: $x=y$ in the model. By $x=y$, the point set could be separated into two part, $x < y$ and $x \geq y$, which is the same to definition of $side$(definition \ref{Side}). Thus, each point can be assigned to one side.
	
	According to the definition \ref{Side}, we can infer that if two points on the same side using highway, they must use the highway on their side. It's easy to prove the correctness since the closet highway of the two points are both the one on their side. Therefore, at most one highway would be used by these two points.

\begin{lemma} \label{the same side use one H}
 	Any two points on the same side use at most one highway.
\end{lemma}
 	
\begin{proof}
Assume there are two point p, q on side-$H_{k}$ (k can be x or y). Because p and q are on side-$H_{k}$, the closet highway to them are both $H_{k}$. So, if p to q using highway, it need to "enter $H_{k}$, leave $H_{k}$", which use only one highway $H_{k}$.
\end{proof}

	By lemma \ref{the same side use one H}, for points on the same side, at most one highway would be used, which is the same to the research of previous works \cite{Aloupis, Dai}. By previous works, the time-convex hull for a given set \textbf{S} of n points in $L_p$-metric can be computed in $O(n\log{n})$ time using $O(n)$ space, where $1 \leq p \leq \infty$. Consequently, clusters could also be constructed in O($n\log{n}$).
	
\begin{corollary} \cite{Aloupis, Dai} \label{cluster nlogn}
	For a point set $S_k$ on $Side-H_k$, clusters could be constructed in O($n\log{n}$).
\end{corollary}
	
	For a point set \textbf{S}, we separate all points into two sides. Points on $Side-H_x$ are sorted by their x-coordinates, and points on $Side-H_y$ are sorted by their y-coordinates. By lemma \ref{the same side use one H} and corollary \ref{cluster nlogn}, we could construct clusters on each side. 
	
	However, the walking regions of clusters we constructed on one side may contain some points on the other side, which means that we may need to merge clusters on different sides. \bigskip
	
	To merge the clusters on different sides, we record the boundary of walking region of each cluster. Take the points on $Side-H_y$ for example, if any point $p$ on $Side-H_y$ lies inside the walking region of any cluster $C_X$ on $Side-H_x$, we would mark $C_X$ and the cluster $C_Y$ that $p$ belongs to. The marked clusters would be merged into a new cluster in subsequent steps. 
	
	After all points on $Side-H_y$ have been checked, we execute the same process to $Side-H_x$ and mark the clusters we found. Finally, these marked clusters would be merged into a new cluster. The cluster merging of different sides is called \textit{\textbf{side inclusion}}.

\begin{figure*}
\begin{center}
	\includegraphics[scale=0.3]{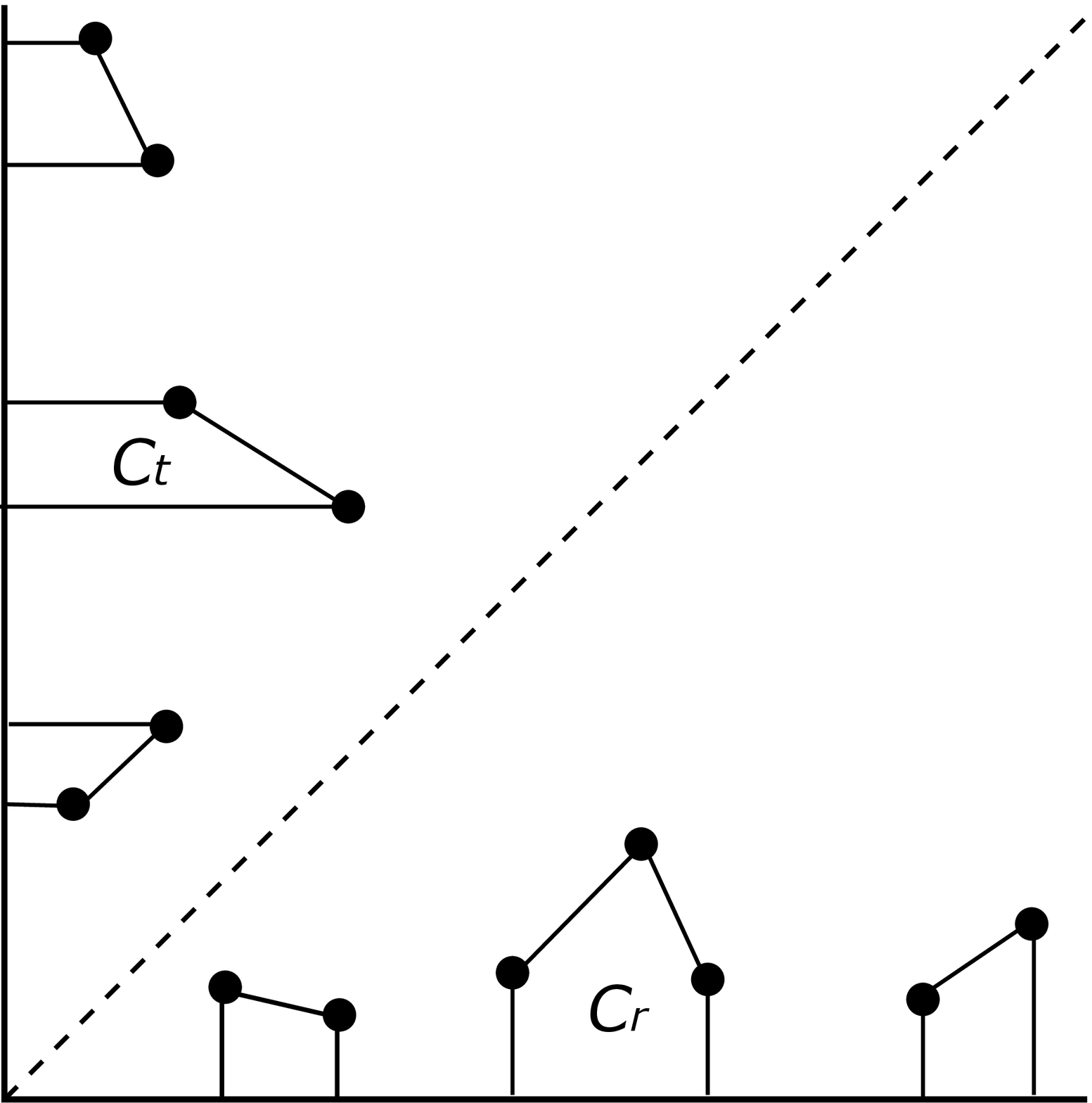} \hspace{10mm}
	\includegraphics[scale=0.3]{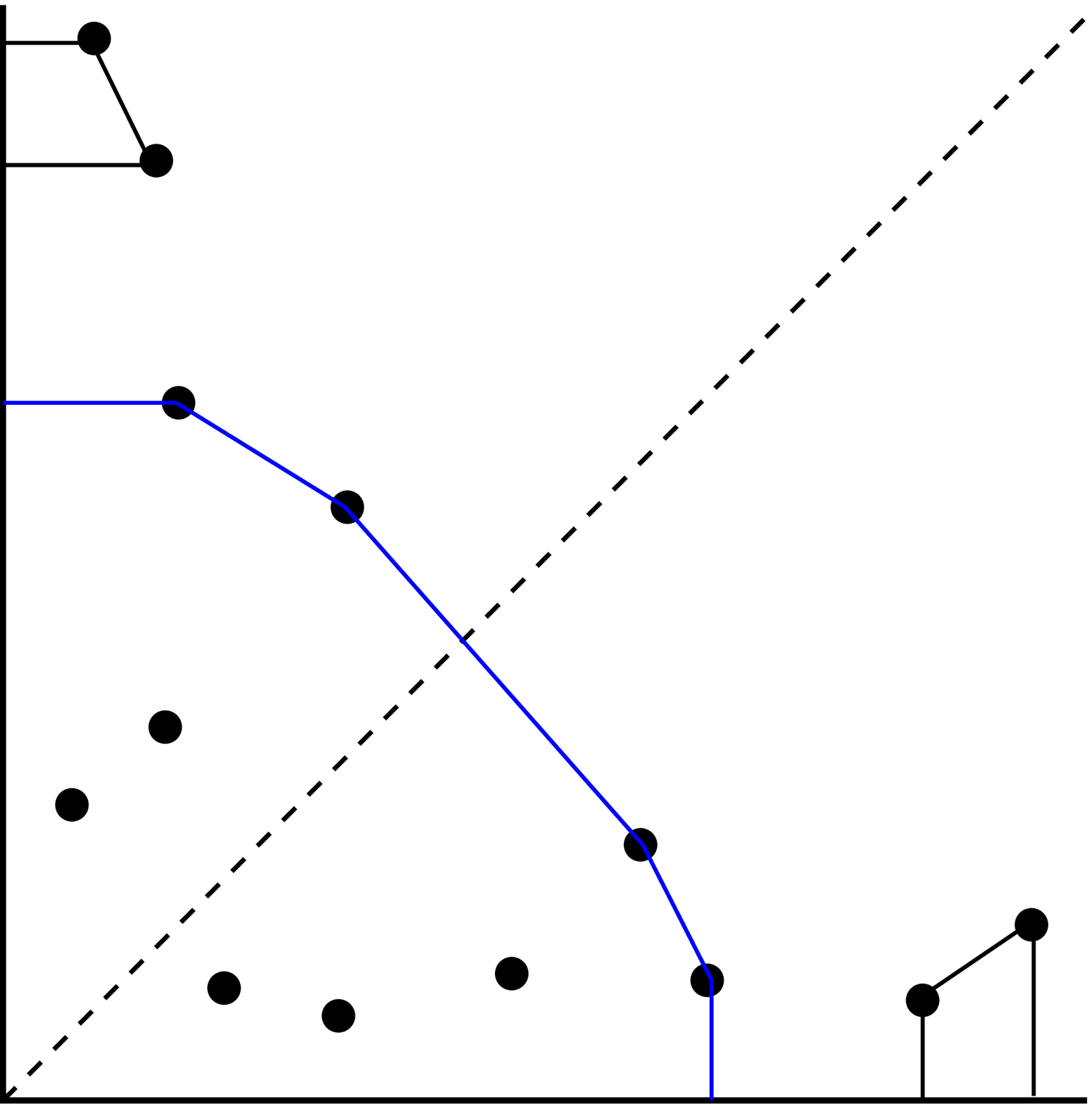} 
	\caption{Every marked clusters are located at bottom left of $\overline{C_t C_r}$}
\end{center}
\end{figure*}

\begin{lemma}[merge approach] \label{merge approach}
	Let $C_t$ be the topmost cluster on $Side-H_y$ and $C_r$ be the rightmost cluster on $Side-H_x$ we marked. Merging $C_t$ and $C_r$ will also merge every clusters on the left of $C_r$ or below of $C_t$.
\end{lemma}

\begin{proof}
	Let $P_t$ be a point of $C_t$ and $P_r$ be a point of $C_r$. $P_t$ and $P_r$ lie in the walking region of each other. Denote $s$ to be the shortest time-path between $P_t$ and $P_r$. Since $P_t$ and $P_r$ lie in the walking region of each other, $s$ doesn't use highway. Suppose that $s$ intersects another cluster at some point $x$. Because of the fact that a subpath of a shortest path is also a shortest path, neither the path $\overline{P_t x}$ nor the path $\overline{P_r x}$ use the highway. Thus, $x$ and one of $P_r$ or $P_t$ must be in the same cluster before merging, which is a contradiction to the assumption.
\end{proof}	

	By lemma \ref{merge approach}, the only marked clusters which need to be merged are the topmost cluster $C_t$ and the rightmost cluster $C_r$. Because merging $C_t$ and $C_r$ will also merge every clusters on the bottom-left of them, the new cluster after merging will contain all the other points before $C_t$ and $C_r$. Thus, all marked clusters have been merged and no more clusters would be merged again. We could end the side inclusion and construct the time-convex hull. 
	
	Although the structure of our algorithm has been introduced, we haven't explained how to find the clusters need to be marked.

	Our algorithm is based on the research that was proposed by Aloupis et al \cite{Aloupis}. In their research, construct a time-convex hull with one highway in the $L_1$-metric takes only $O(n\log{n})$. Therefore, by corollary \ref{cluster nlogn}, we could construct clusters on each side by the method of Aloupis et al in $O(n\log{n})$.

	Point inclusion is an approach to construct clusters with one highway, which was also proposed by Aloupis et al \cite{Aloupis}. Take $Side-H_x$ for example, we would sort the point set $S_x$ by their x-coordinates first, and check the walking region of each point one by one. The walking regions are saved as linked lists, which record the boundaries of the walking region for each point. If a new point $P_i$ lies in the walking region of $P_{i-k}$, then all points between $P_{i-k}$ and $P_i$ would be merged into a new cluster. After solving $P_i$, we will turns to do point inclusion to $P_{i+1}$, until every points in $S_x$ are solved. 
	
	The algorithm saves the boundaries of walking regions when constructing clusters, so we can check if any points on the other side lie on the walking region at the same time.
	
\begin{definition}[Segment-dragging query] \label{segment-dragging}
The segment-dragging query, denoted Q(L). Q(L) ask that, for any line segment L with finite slope and a point set P, if P$\cap L^{+}$ is empty or not, where $L^{+}$ is the half plane to the right of L or above L.

We call the query searching the half plane to the right of L \textbf{segment horizontal-dragging query}, while searching the half plane above L is called \textbf{segment vertical-dragging query}
\end{definition}

\begin{lemma} \label{complexity segment-dragging}
(Chazelle \cite{Chazelle} and Mitchell \cite{Mitchell}) Segment-dragging query takes O($n\log{n}$) constructed time, O($n$) space, O($\log{n}$) query time.
\end{lemma}
	
	By property \ref{WR_L1}, for the incoming point $p_i$ on $Side-H_x$, its left boundaries of walking region consist of at most three segments on $Side-H_y$ (See fig. \ref{fig:WR L1}), while the right boundary consist only one vertical segment on $Side-H_y$. Therefore, we extend the right boundary to a vertical line $R$, and invoke the segment horizontal-dragging query for each left boundaries to $R$. If the query return yes, it means that the cluster of $p_i$ must be merged with some clusters on the other side, so the cluster $p_i$ belongs to would be marked. 
	
\begin{lemma} \label{at most 3 segment-dragging}
	Each point evokes at most 3 segment-dragging queries.
\end{lemma}	
	
	However, checking points is not enough. As fig. \ref{fig:edge_inclusion}, the cluster that contains $p_{1}$ and $p_{2}$ should also contain $p_{3}$, but $p_{3}$ is neither in the walking region of $p_{1}$ or $p_{2}$. The cause of the problem is that the walking region of a cluster includes not only the walking regions of its convex vertices, but also the walking regions of the line segments that connect the convex vertices.
	
\begin{figure}[H]
\begin{center} 
	\includegraphics[scale=0.7]{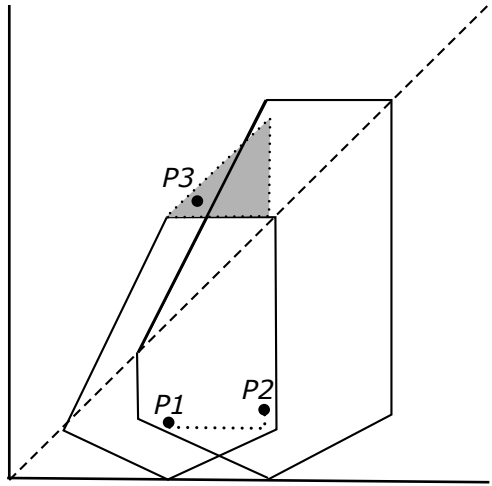} 
	\caption{Edge inclusion} \label{fig:edge_inclusion}
\end{center}
\end{figure}

	To deal with the problem, we need to check whether the walking region of the line segments(gray area in fig. \ref{fig:edge_inclusion}) causes side inclusion, which is called \textit{\textbf{edge inclusion}}. The edge of a cluster would only be formed when clusters are merged. Whenever a new edge is formed, the segment-dragging query would be invoked to check if the walking region of the edge contains any points on the other side. If the query return yes, mark the cluster which the edge belongs to. \bigskip
	
	After solving all the points on $Side-H_x$, we perform the same process for the points on $Side-H_y$.
	
	For clusters we marked above, we pick out the topmost cluster $C_t$ and the rightmost cluster $C_r$. By lemma \ref{merge approach}, merging $C_t$ and $C_r$ will also merge all the other clusters before them. Therefore, we can get the final clusters after merging $C_t$ and $C_r$. \bigskip
	
	After we obtain the final clusters, the time-convex hull could be constructed by linear scan in O($n$). The linear scan start from the last point $P_{x_n}$ on $Side-H_x$ to the fist point $P_{x_1}$, then turn to $Side-H_y$ from the first point $P_{y_1}$ to the last point $P_{y_n}$.
	
	For clusters using only one highway, details steps have been record in the work of Yu and Lee \cite{Yu} and Aloupis et al \cite{Aloupis}. By lemma \ref{merge approach}, after side inclusion, there is no more than one cluster using two highways, and the cluster must be the first cluster on both $Side-H_x$ and $Side-H_y$. Therefore, we can scan all points in the cluster linearly, and the details would be the same to the method of constructing clusters using only one highway. \bigskip
	
\begin{theorem}
	The time-convex hull with two orthogonal highways under $L_1$-metric can be computed in $O(n\log{n})$ time.
\end{theorem}

\begin{proof}
	By corollary \ref{cluster nlogn}, clusters on each sides could be constructed in $O(n\log{n})$ times.
	
	For a point set contains n points, by lemma \ref{at most 3 segment-dragging}, each point would evoke at most 3 segment-dragging queries. Thus, the segment-dragging query would be evoked at most $O(n)$ times for n points. Since there are at most n clusters, and each merging create only one newly-formed edge, the segment-dragging query would be evoked at most $O(n)$ times for n edges. By lemma \ref{complexity segment-dragging}, segment-dragging query takes $O(\log{n})$ times. To sum up, the side inclusion tests totally take O($n\log{n}$) times.
	
	In the last step, it takes $O(n)$ times to construct the time-convex hull by linear scan.
	
	Therefore, the time-convex hull with two orthogonal highways under $L_1$-metric can be computed in $O(n\log{n})$ time.
\end{proof}

\section{Optimal Construction in $L_2$-metric with infinite highway speed}

\hspace{4mm} 
In this section, we just care about the case of infinite highway speed. By lemma \ref{merge approach}, our approach is to find the topmost marked cluster $C_t$ and the rightmost marked cluster $C_r$ as well. For this goal, we need to construct clusters first and check where $C_t$ and $C_r$ are. Like what we do under $L_1$-metric, we process the point inclusion to every points. However, due to the different shape of walking regions between $L_1$ and $L_2$ metric, edge inclusions may also be happened during point inclusion. The method of constructing clusters with one highway under $L_2$-metric was proposed by Dai \cite{Dai}, which takes $O(n\log{n})$ time.\bigskip

\begin{figure*}
\begin{center}
\subfloat[edge inclusion on the same side]{\includegraphics[scale=0.3]{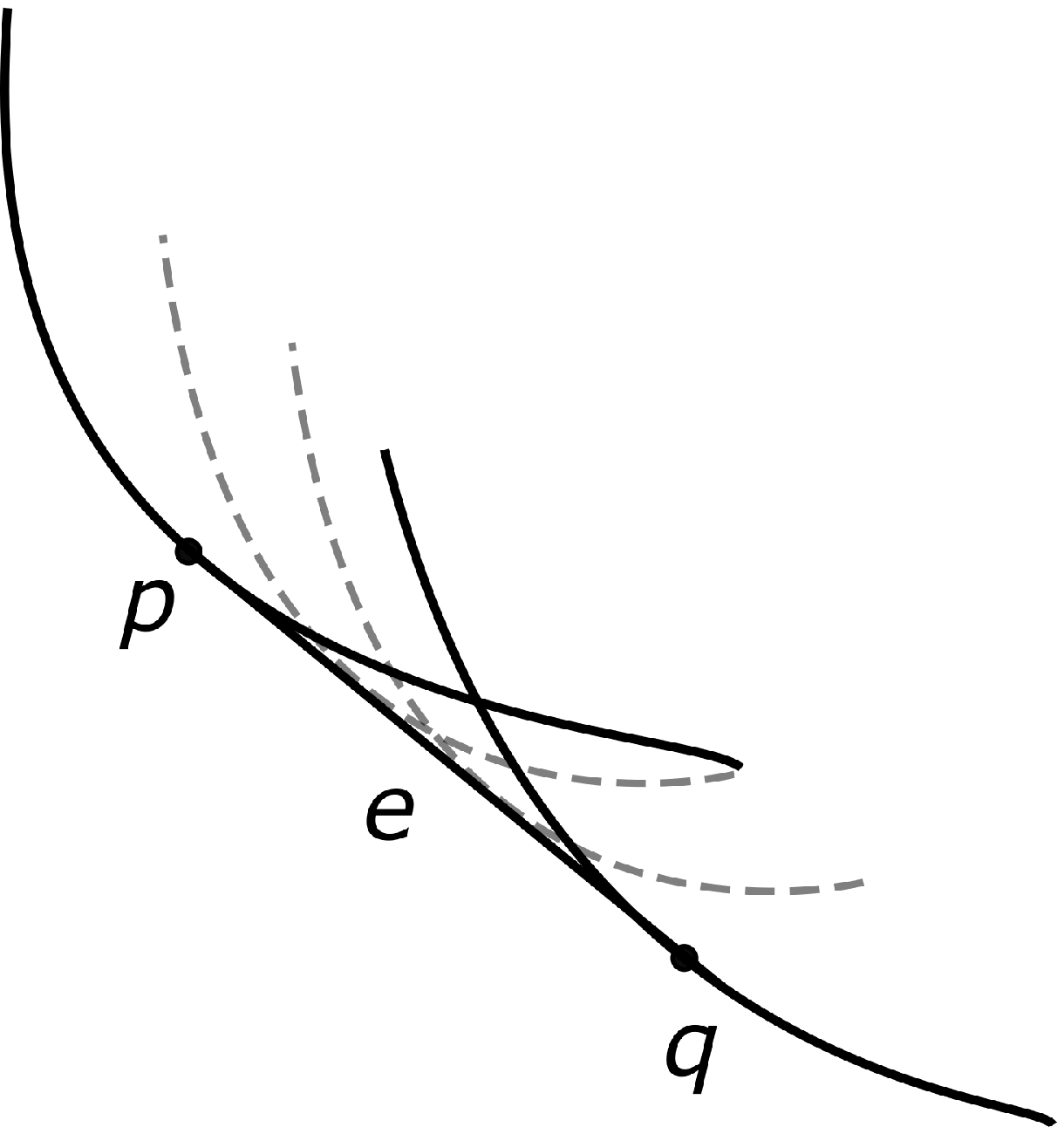}}
\hspace{4mm}
\subfloat[new-formed edge e]{\includegraphics[scale=0.3]{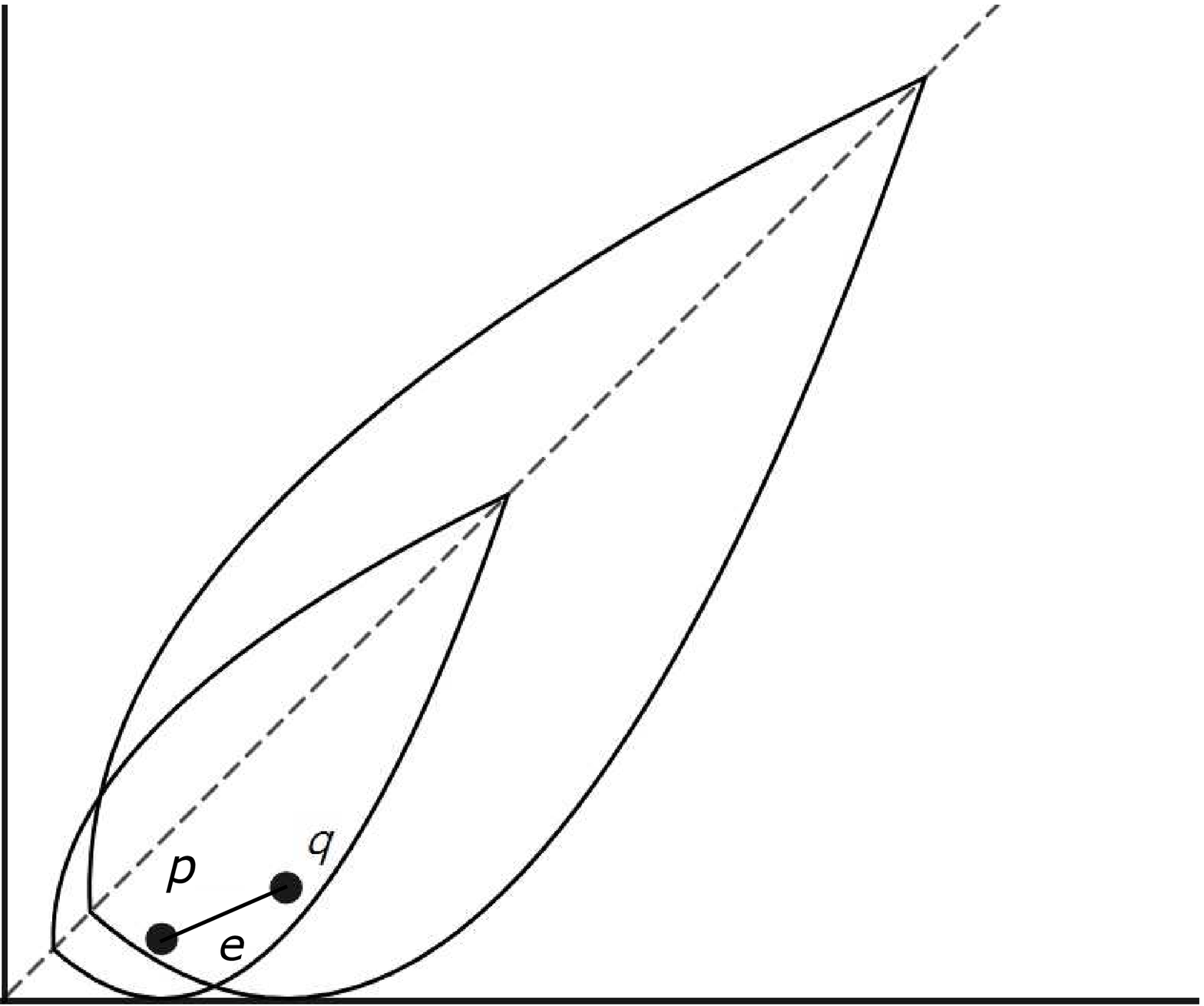}}
\hspace{4mm}
\subfloat[edge inclusion on the other side]{\includegraphics[scale=0.3]{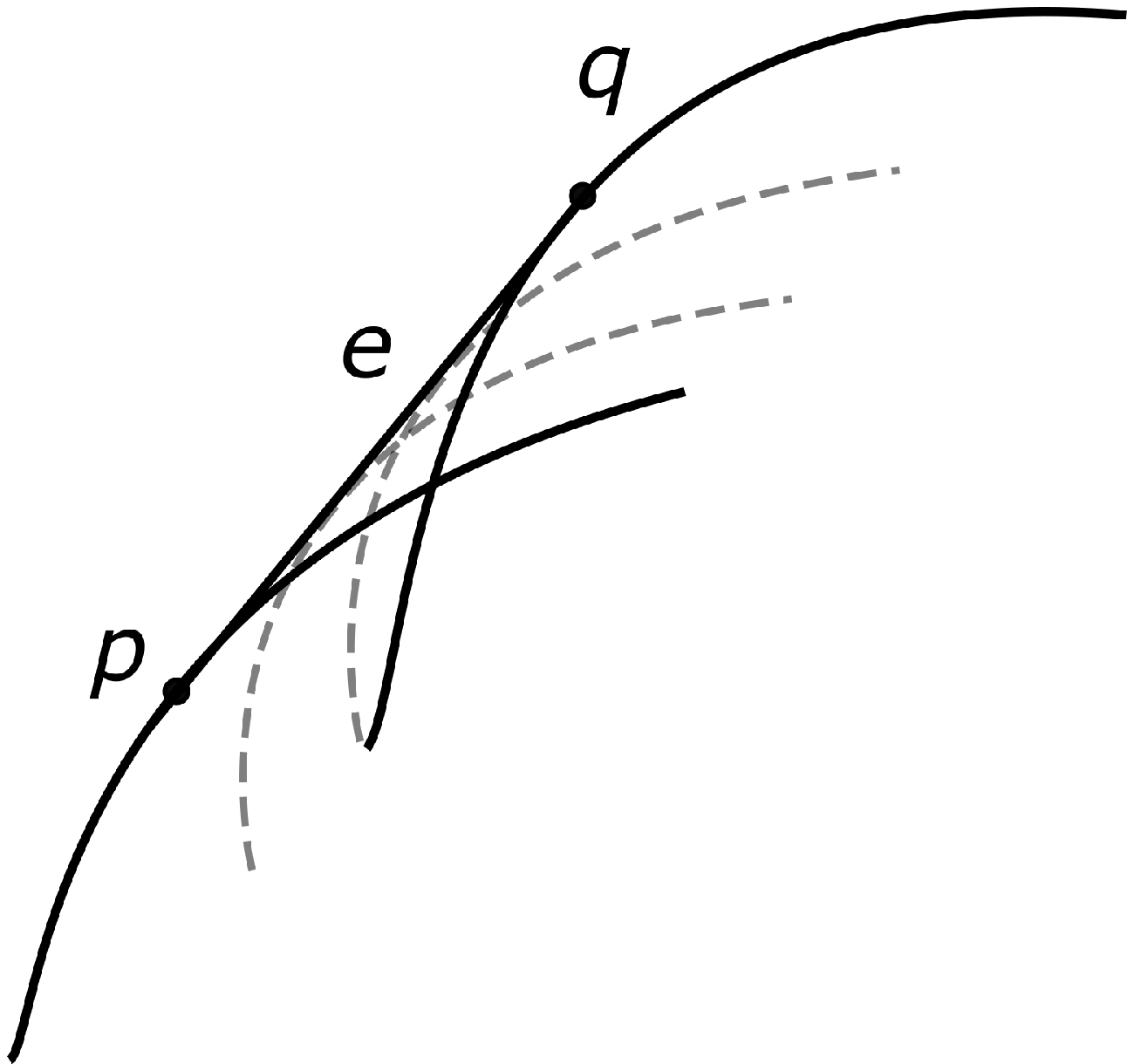}}
\caption{edge inclusion} \label{fig:edge_inclusion_L2}
\end{center}
\end{figure*}

	The edge inclusion will be happened when the clusters are merged. When the clusters are merged, a new hull edge will be created (See Fig. \ref{fig:edge_inclusion_L2}(b)), and we exploit segment vertical-dragging query for this test. The approach we used is referred to the work of Dai \cite{Dai}. Let \textbf{E} be the set of newly-formed edges. We randomly pick one edge $e$ in \textbf{E} to check if there are any points that lie in this area $WR(e)$, which means that the clusters need to be merged. Otherwise we check another edge in \textbf{E} until \textbf{E} is empty. This step just check the clusters on the same side.

	On the other side, edge inclusion may also occur(See Fig. \ref{fig:edge_inclusion_L2}(c)). Thus the walking region of edge \textbf{e} on the other side also need to be checked. We invoke segment horizontal-dragging query again to check if $WR(e)$ on the other side contain any points. If the query return yes, this cluster will be marked. \bigskip

	After the edge inclusion test of incoming point $p_i$, we need to maintain the outermost boundary of clusters on the other side, i.e. if $p_i$ is on $Side-H_x$, we would maintain the outermost boundaries of clusters on $Side-H_y$ (See Fig.\ref{fig:boundary}) . This process will be done in both side. The boundary will be stored as linked-lists which consist of two points and a curve. Below we describe how to maintain the boundary.

When a new point $q_i$ arrives, we will check the intersection between $WR(q_i)$ and $x=y$. Denoted these two points $b_j$ and $b_k$, where $b_j$'s y-coordinate is less than $b_k$'s y-coordinate :\\
1. If $b_j$'s and $b_k$'s y-coordinate are both larger than the tail of linked-list $L_k$, we can update the linked-list directly.\\
2. If $b_j$'s and $b_k$'s y-coordinate are both less than $L_k$, the walking region will be covered and we can go next point $q_{i+1}$.\\
3. After case 1 and 2, we check if the walking region of $q_i$ intersects with the boundary which is represented by $L_k$.

\hspace{4mm} If yes, update the linked list and go for next point $q_{i+1}$.

\hspace{4mm} If no, $L_k$ will be covered and never be used, we can delete it and keeping going to $L_{k-1}$ 

\hspace{4mm} to check these three cases. \bigskip

\begin{figure*}
\begin{center}
	\includegraphics[scale=0.3]{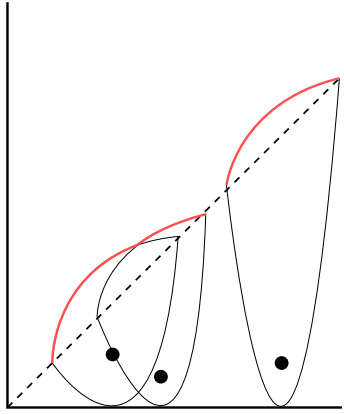} 
	\caption{The boundary of walking region of points	 in Side-$H_x$} \label{fig:boundary}
\end{center}
\end{figure*}

We construct the linked-list of points on $Side-H_x$ first. After the boundary is constructed, ray shooting could be executed to compare each point on $Side-H_y$ with the linked list which is made of walking regions from $Side-H_x$. The ray shooting would start from the last point $P_{y_n}$ of $Side-H_y$ and the tail of linked list $L_k$. \bigskip

If $L_k$ is to the left of $P_{y_n}$, it means that $P_{y_n}$ lies in the walking region of some clusters on $Side-H_x$. Therefore, we could mark the cluster that $P_{y_n}$ belongs to and stop the ray shooting on $Side-H_y$. Otherwise, when $L_k$ is not to the left of $P_{y_n}$ but to the right of it, we could make sure that $P_{y_n}$ doesn't lie in any walking region of clusters on $Side-H_x$ for the fact that the linked list is y-monotone. Thus we should search the next point $P_{y_{n-1}}$.
	
	When the y-coordinate of $P_{y_n}$ is less than $L_k$, we traverse the next curve $L_{k-1}$. Like the concept of the former, we would search next point $P_{y_{n-1}}$ if the y-coordinate of $L_k$ is less than $P_{y_n}$. The ray shooting will be executed continuously until all points on $Side-H_y$ have been searched or all curves in the linked list have been traversed.\bigskip
	
	After we have traversed every points, $C_t$ and $C_r$ could be found and we can complete the final clusters after merging them. If there are no marked cluster in either $Side-H_x$ or $Side-H_y$, it means that the cluster-merging between different sides would not happen, which is the same to the case of one highway that has been solved efficiently by Dai \cite{Dai}.\bigskip
	
	In the point inclusion test, it can be done in $O(n\log{n})$ by Dai \cite{Dai}. Furthermore, in the edge inclusion test, every cluster-merging create one hull edge at most and take $O(\log{n})$ in each test. Therefore, it takes $O(n\log{n})$ in edge inclusion test. In the process of maintain the outermost boundary, each point intersect at most two points and every nodes are just traversed at most one time, so it takes $O(n)$ time in ray shooting test. Thus, the time-convex hull with two orthogonal highways of infinite speed under $L_2$-metric can be computed in $O(n\log{n})$ time.
	
\begin{theorem}
	The time-convex hull with two orthogonal highways of infinite speed under $L_2$-metric can be computed in $O(n\log{n})$ time.
\end{theorem}

\section{Conclusion}
\hspace{4mm} In this paper, we give an O($n\log{n}$) time algorithm for the time-convex hull with two orthogonal highways under $L_1$-metric. For $L_2$-metric, we provide O($n\log{n}$) time algorithm in the special case where the highway speed is infinite. There are some extensions of this work, e.g., $L_2$-metric with general highway speed, $L_p$-metric where 1 $\leq p \leq \infty$, or, with multiple highways.


\bibliographystyle{plain}
\bibliography{highway-hull-$L_p$-metric}

%
%
%



\end{CJK*}
\end{document}